\documentclass[twocolumn]{svjour3}

\smartqed

\usepackage{graphicx}
\usepackage{mathptmx}

\usepackage{etex}
\usepackage{algorithm}
\usepackage{algorithmicx}
\usepackage{algpseudocode}
\usepackage{amsfonts}
\usepackage{amsmath}
\usepackage{amssymb}
\usepackage{amstext}

\usepackage{appendix}
\usepackage{array}
\usepackage{booktabs}
\usepackage{bm}
\usepackage{cancel}
\usepackage{cases}
\usepackage{color}
\usepackage{colortbl}
\usepackage{dsfont}
\usepackage{enumerate}
\usepackage{fancyhdr}
\usepackage[perpage]{footmisc}
\usepackage{float}
\usepackage{framed}
\usepackage{lastpage}
\usepackage{latexsym}
\usepackage{listings}
\usepackage{lmodern}
\usepackage{mathtools}
\usepackage{mathrsfs}
\usepackage{multicol}
\usepackage{multirow}

\usepackage{pst-all}
\usepackage{relsize}
\usepackage{setspace}
\usepackage{stmaryrd}
\usepackage{subfig}
\usepackage{tabularx}
\usepackage{tikz}
\usepackage{upgreek}
\usepackage{url}
\usepackage{verbatim}
\usepackage{wasysym}
\usepackage{xcolor}
\usepackage{xspace}

\usepackage{authblk}

\graphicspath{{images/}}

\DeclareMathOperator{\Prob}{{\mathbb P}}
\DeclareMathOperator{\Exp}{{\mathbb E}}

\DeclareMathOperator{\Var}{\mathbb{V}\mathrm{ar}}

\newcommand*{\NormDist}{\mathsf{Normal}}
\newcommand*{\LNDist}{\mathsf{Lognormal}}

\newcommand*{\WeibullDist}{\mathsf{Weibull}}
\newcommand*{\ParetoDist}{\mathsf{Pareto}}

\newcommand*{\ClaytonCop}{\mathsf{Clayton}}
\newcommand*{\GumbelCop}{\mathsf{GumbelHougaard}}
\newcommand*{\FrankCop}{\mathsf{Frank}}

\newcommand*{\convdistr}{\xrightarrow{\:\smash{\mathcal{D}}\:}}

\newcommand*{\iidDist}{\overset{\smash{\mathrm{iid}}}{\sim}}
\newcommand*{\indDist}{\overset{\smash{\mathrm{ind}}}{\sim}}

\newcommand*{\NL}{\mathbb{N}_+}
\newcommand*{\RL}{\mathbb{R}}

\newcommand*{\dd}{\mathop{}\!\mathrm{d}}
\newcommand*{\e}{\mathrm{e}}

\DeclareMathOperator{\csch}{csch}

\newcommand*{\ind}[1]{\mathbb{I}\{ #1 \}}

\newcommand*{\Lp}{\mathcal{L}}

\renewcommand{\hat}[1]{\widehat{#1}}

\renewcommand{\epsilon}{\varepsilon}
\renewcommand{\pi}{\uppi}

\newcommand*{\ih}{\mathrm{i}}

\newcommand*{\remQED}{{\mbox{\, \vspace{3mm}}} \hfill \mbox{$\Diamond$}}

\newcommand*{\Ftail}{\mkern 1mu\overline{\mkern-1mu{F}\mkern+1.25mu}}

\newcommand*{\for}[1]{\,,\qquad \text{for } #1}

\definecolor{Gray}{gray}{0.9}

\newcommand*{\bfI}{\bm{I}}

\newcommand*{\bfW}{\bm{W}}
\newcommand*{\bfX}{\bm{X}}
\newcommand*{\bfY}{\bm{Y}}

\newcommand*{\bfe}{\bm{e}}

\newcommand*{\bft}{\bm{t}}

\newcommand*{\bfx}{\bm{x}}

\newcommand*{\bfmu}{\bm{\mu}}

\newcommand*{\bfSigma}{\bm{\Sigma}}
\newcommand*{\bftheta}{\bm{\theta}}
\newcommand*{\bfTheta}{\bm{\Theta}}

\newcommand*{\bfzero}{\bm{0}}
\newcommand*{\bfone}{\bm{1}}

\newcommand{\Asymf}{\mathfrak{f}_S}

\newcommand{\AsymFTail}{\overline{\mathfrak{F}}_S}

\newcommand{\AsymfTrun}{\mathfrak{f}_{S|S>\gamma}}

\renewcommand{\phi}{\varphi}
\renewcommand{\rho}{\varrho}

\journalname{Statistics and Computing}

\begin{document}

\title{Rare tail approximation using asymptotics and $L^1$ polar coordinates}

\author{Thomas Taimre \and Patrick J.\ Laub}

\institute{T.\ Taimre \at
	School of Mathematics and Physics \\
	University of Queensland, Brisbane \\
	\email{t.taimre@uq.edu.au}
	\and
	P.\ J.\ Laub \at
	The University of Queensland, Brisbane \& \\
	Aarhus University, Aarhus \\
	\email{p.laub@[uq.edu.au$\mid$math.au.dk]}
}

\date{Received: \today / Accepted: date}

\maketitle

\begin{abstract}
In this work, we propose a class of importance sampling (IS) estimators
for estimating the right tail probability of a sum of continuous random variables
based on a change of variables to $L^1$ polar coordinates in which the radial and angular components
of the IS distribution are considered separately.

When the asymptotic behaviour of the sum is known we exploit it for the radial change of
measure, and the resulting estimator has the appealing form of the (known) asymptotic
multiplied by a random multiplicative correction factor.
Given we assume knowledge of the asymptotic behaviour of the sum in this framework, traditional
notions of efficiency that appear in the rare-event literature hold little practical meaning here.
Instead, we focus on the practical behaviour of the proposed estimator in the pre-asymptotic regime
for right tail probabilities between roughly $10^{-3}$ and $10^{-7}$.

The proposed estimator and procedure are applicable in both the heavy- and light-tailed
settings, as well as for independent and dependent summands.
In the case of independent summands, we find that our estimator compares favourably with
exponential tilting (iid light-tailed summands) and the
Asmussen--Kroese method (independent subexponential summands).

However, for dependent subexponential summands using the same simple angular distribution as for the independent case,
the performance of our estimator rapidly degenerates with increasing dimension, suggesting an open avenue for further research.
 \end{abstract}

\section{Introduction}

A typical problem in the field of rare-event estimation is to determine the probability
\begin{equation} \label{ProbabilityOfInterest}
\ell(\gamma) := \Prob(S > \gamma)
\end{equation}
where $S := X_1 + \dots + X_d$ for a random vector $\bfX:=(X_1,\dots,X_d)$ with
fixed $d \in \NL$
having joint probability density function (pdf) $f_{\bfX}$
and where the $\gamma \in \RL$ is large or increasing. In applications, we often wish to understand the behaviour of a combination of random factors, and hence the random variable (r.v.) $S$ is ubiquitous in real-world modeling problems. It can model, for example: aggregate risk or portfolio value for holding $d$ risky assets \cite{mcneil2015quantitative,Rueschendorf2013}, the aggregate losses for $d$ insurance policy claims \cite{asmussen2010ruin,klugman2012loss}, or the combined signal interference from $d$ wireless transmission sources \cite{fischione2007approximation}. Probabilities of the form \eqref{ProbabilityOfInterest} are used to understand how a system would behave under extreme scenarios such as a market crash, a power surge, or a natural disaster. One is typically interested not just in the quantity $\ell(\gamma)$ but also in the behaviour of the summands when the extreme event $\{ S > \gamma \}$ occurs.

This probability is available in closed-form for only a few basic cases, when the density of $S$ (which is a $d$-fold convolution) has a known solution, c.f. \cite{nadarajah2008review}. For example, when the summands are independent and identically distributed (iid) then it is sometimes simple to calculate (for example exponential, gamma, or normal summands, and in the discrete case binomial, geometric, or negative binomial summands)
and sometimes it is still intractable
(for example lognormal, Weibull, Laplace, or Beta summands). However, requiring the assumption of independence (let alone iid-ness) of the summands is a stifling restriction when modeling real-world events; a notorious example would be the partial blame of the 2008--9 global financial crisis on mathematicians' inappropriate use of a simplistic dependence model (the Gaussian copula) \cite{salmon2009recipe}.

When analytical solutions are unavailable, the next best option is numerical integration, and after that Monte Carlo integration (or quasi-Monte Carlo).
Numerical integration algorithms applied to
\[ \ell(\gamma) = \int_{\RL^d} \ind{ x_1 + \dots + x_d > \gamma} f_{\bfX}(\bfx) \dd \bfx \,,\]
where $\ind{A}$ denotes
the indicator of an event $A$ (taking value 1 if $A$ occurs and 0 otherwise),
are typically slow, inaccurate, and misleading. This is because the indicator is rarely 1, floating-point errors accumulate, and the curse of dimensionality applies for $d$ larger than about 2 or 3. Some of these algorithms attempt to estimate the error in their result, but there are few (if any) theoretical guarantees that these estimates are reliable.

Rare-event problems also cause difficulties for the crude Monte Carlo (CMC) estimator. This is obvious as the CMC estimator's relative error explodes for large $\gamma$ --- that is, the CMC estimator $\hat{\ell}_{\text{CMC}}(\gamma) := \ind{S > \gamma}$ has
\begin{align*}
&\hphantom{=}~\lim_{\gamma \to \infty} \text{RelativeError}\{ \hat{\ell}_{\text{CMC}}(\gamma) \} \\
&= \lim_{\gamma \to \infty} \frac{ \Var[ \hat{\ell}_{\text{CMC}}(\gamma) ] }{ \ell(\gamma)^2 }
= \lim_{\gamma \to \infty} \frac{  \ell(\gamma)[1-\ell(\gamma)] }{  \ell(\gamma)^2 } = \infty \,.
\end{align*}
Intuitively, the problem is because the indicator $\ind{S > \gamma}$ is eventually always 0 when $\gamma$ gets very large. In response, various variance reduction techniques have been applied so that there are now a large collection of estimators with better performance in this setting, c.f. `rare-event estimation' in \cite{kroese2013handbook,asmussen2007stochastic,glasserman2003monte}.

There is, of course, no silver bullet for the problem. Some estimators only apply to specific distributions (e.g. \cite{botev2017fast} for sums of lognormals, \cite{yao2016estimating} for sums of phase-type mixtures) or to certain classes of distributions (exponential tilting for light-tailed summands \cite{kroese2013handbook,asmussen2007stochastic}, hazard-rate twisting or the Asmussen--Kroese method \cite{asmussen2006improved} for heavy-tailed summands). Other estimators are general but require specifying either some extra information (e.g.\ availability of conditional distributions for conditional Monte Carlo \cite{asmussen2017conditional}, or an appropriate sampling distribution for use in importance sampling). The most general estimators --- such as the generalised splitting method, cross-entropy method, or Markov Chain Monte Carlo (MCMC) methods such as  \cite{chan2012improved} --- are usually computationally demanding, they often depend upon an intelligent selection of input parameters to perform efficiently, and are somewhat complicated.

Whilst one rarely has an exact expression for $\ell(\gamma)$, it is somewhat common to know an \emph{asymptotic approximation} to it, and this forms the basis for our proposed estimator. For example, if $\bfX \sim \mathsf{Lognormal}(\bfmu, \bfSigma)$  where $\bfmu \in \RL^d$ and $\bfSigma \in \RL^{d \times d}$ is positive definite (by which we mean that $\bfX=\exp(\bfY)$ component-wise, where $\bfY$ has
a multivariate normal distribution with mean vector $\bfmu$ and covariance matrix $\bfSigma$), then it has been shown that \cite{asmussen2008asymptotics}
\begin{equation} \label{SLNasymptotic}
\ell(\gamma) = \Prob(S > \gamma) \sim \sum_{i=1}^d \Prob(X_i > \gamma) \qquad \text{ as } \gamma \to \infty
\end{equation}
where $f(x) \sim g(x)$ denotes $\lim_{x \to \infty} f(x)/g(x) = 1$. Thus, one is tempted to label the right-hand side (RHS) of \eqref{SLNasymptotic} as $\hat{\ell}_{\mathrm{Asym}}(\gamma)$ and use it as an approximation for $\ell(\gamma)$. For certain values of $(\bfmu,\bfSigma)$ this asymptotic approximation can be accurate, in others it can be wildly inaccurate, depending on how fast the asymptotic approximation converges to the true value; see Figure~\ref{fig:slow_convergence} for an illustration where it is only when $\ell(\gamma) \lesssim 10^{-10}$ that the asymptotic form begins to give accurate estimates (i.e., $\hat{\ell}_{\mathrm{Asym}}(\gamma) / \ell(\gamma) > 0.99$). A discussion of this phenomenon is in \cite{botev2017fast}.

\begin{figure}[h]
	\centering
	\includegraphics[width=0.49\textwidth]{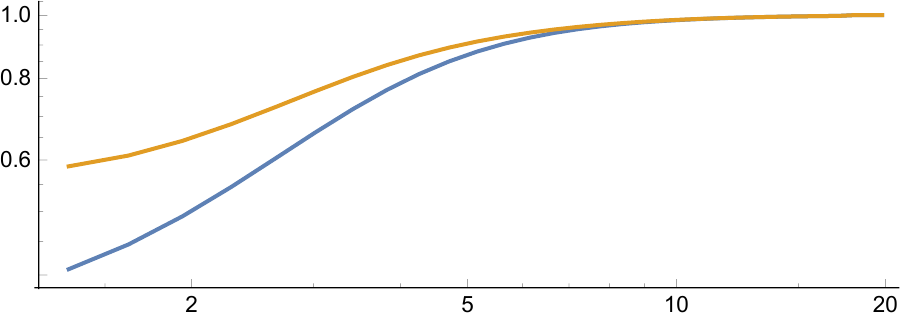} \\
	\includegraphics[height=0.85em]{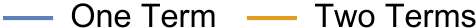}
	\caption{A comparison of $\ell(\gamma)$ and $\hat{\ell}_{\text{Asym}}(\gamma)$ for $X_1+X_2$ where $X_1 \sim \LNDist(0, 1)$ is independent to $X_2\sim\LNDist(0,\frac{3}{4})$. The $y$ axis plots $\hat{\ell}_{\text{Asym}}(\gamma) / \ell(\gamma)$, and the $x$ axis shows ${-}\log_{10} \ell(\gamma)$. The two curves describe two possible asymptotics: the yellow ``Two terms'' describes $\hat{\ell}_{\text{Asym}}(\gamma)$ as given in \eqref{SLNasymptotic}, whereas the blue ``One Term'' uses just the (eventually dominant) first term of this sum.}
\label{fig:slow_convergence}
\end{figure}

We propose an \emph{importance sampling} (IS) estimator which incorporates the asymptotic approximation and uses Monte Carlo sampling to estimate a correction to $\hat{\ell}_{\mathrm{Asym}}(\gamma)$
in order to construct an unbiased estimator $\hat{\ell}(\gamma)$ of  $\ell(\gamma)$.
The main drawback to IS is \emph{likelihood degeneration}, where one can face numerical errors if $\gamma$ or $d$ is extremely large.
The degeneration caused by a large $d$ is only partially compensated by our approach, so we take $d \le 100$.
To mitigate degeneration for large $\gamma$, we focus our attention of values of $\gamma$ which are moderately large but not unrealistically so.
Our goal is to provide an estimator which is practically useful when $\ell(\gamma)$ is between roughly $10^{-3}$ and $10^{-7}$.

The range of probabilities that we consider are unusual as they are less rare than much of the standard rare-event literature. The orthodox approach is to construct an estimator $\hat{\ell}(\gamma)$ and analyse the limit
\[ \lim_{\gamma\to\infty} \Var(\hat{\ell}(\gamma))/\ell(\gamma)^2 \,;\] if the limit is small (i.e., zero, bounded, or at least grows only at a polynomial rate) then the estimator is branded as a success (it has `vanishing relative error', `bounded relative error', or is `logarithmically efficient' respectively) regardless of its behaviour in the finite $\gamma$ situation. It can happen that these desirable limiting properties are only discernible in cases when the probabilities are truly minuscule (e.g. of order $10^{-10}$ or smaller); in a situation such as this, the model error would surely dominate any estimation error.

The remainder of this paper is structured as follows.
The estimator is introduced in Section~\ref{scn:Estimator}, the results from numerical comparisons are in Section~\ref{Sec:Results}, and Section~\ref{Sec:Conclusion} concludes the discussion.

\section{The $L^1$ polar estimator} \label{scn:Estimator}

\subsection{The general form}

We construct an estimator of the quantity $\ell(\gamma) := \Prob(S > \gamma)$, where $S = X_1 + \dots + X_d$ for large $\gamma$ by applying IS.
Standard IS theory says to construct an estimator which samples from a distribution close to $f_{\bfX \,|\, S>\gamma}$ (that is, the distribution of $\bfX$ conditioned on $S>\gamma$), rather than the original $f_{\bfX}$. To do this, we perform a change of variables to {\em Pickand's coordinates} \cite{FalkReiss2005} so
\[ \bfX \longrightarrow (S,\bfTheta) := \left(X_1+\dots+X_d, \bfX / [X_1+\dots+X_d] \right) \,. \]
The new density $f_{(S,\bfTheta)}$ is available (if $f_{\bfX}$ is known), and is
\[
f_{(S,\bfTheta)}(s,\bftheta) = f_{\bfX}(s\bftheta) \times |s|^{d-1}\,.
\]

Consider IS in this new form. Imagine that we have a density $g_{(S,\bfTheta)}$ which is in some way similar to $f_{(S,\bfTheta)}$, for which we also know the marginal density $g_S(s) := \int g_{(S,\bfTheta)}(s,\bftheta) \dd \bftheta$ and the conditional density $g_{\bfTheta|S}:=g_{(S,\bfTheta)}/g_S$. If we truncate $g_{(S,\bfTheta)}$ so that $S>\gamma$ a.s., and use this as the IS distribution, we obtain
\begin{equation} \label{eq:general_form_estimator}
\hat{l}_{\mathrm{IS}}(\gamma) := \frac{\overline{G}_S(\gamma)}{R} \sum_{r=1}^R \frac{ f_{(S,\bfTheta)}(S^{[r]}, \bfTheta^{[r]}) }{ g_S( S^{[r]} ) g_{\bfTheta \mid S}(\bfTheta^{[r]}\,|\, S^{[r]})}
\end{equation}
for $S^{[r]} \iidDist g_{S \mid S>\gamma}$, $\bfTheta^{[r]} \indDist g_{\bfTheta \mid S}( \,\cdot\, | S^{[r]})$,
where we define $\overline{G}_S(\gamma) :=  \int_{\gamma}^\infty g_S(s) \dd s$ and $g_{S \mid S>\gamma}:=g_S \ind{S>\gamma} / \overline{G}_S(\gamma)$.

We investigate estimators of the general form of \eqref{eq:general_form_estimator} which we call \emph{($L^1$) polar estimators}.
These are accurate when $g_{(S,\bfTheta)}= g_S \times g_{\bfTheta|S}$ closely resembles $f_{(S,\bfTheta)}=f_S \times f_{\bfTheta|S}$.
This is carried out in two steps: (i) by finding a \emph{radial approximation} $g_S$ which approximates $f_S$, and (ii) an \emph{angular approximation} $g_{\bfTheta|S}$ similar to $f_{\bfTheta|S}$, which we discuss in the following sections.

\subsection{The radial approximation}

As mentioned in the introduction, we consider utilising an asymptotic form of the sum in our estimator --- they form our radial approximation. To clarify the notation, we precisely define the relevant asymptotic forms:

\begin{definition}[Asymptotic form]
If for some function $\Asymf \in \Lp^1(\RL)$, with tail $\AsymFTail(s) = \int_s^\infty \Asymf(x) \dd x$, and constant $c_S \in \RL_+$, we have that
\begin{equation} \label{asymptotic_sum}
  f_S(s)\ \sim\ c_S \Asymf(s) \for{s \to \infty}
\end{equation}
then we say $\Asymf$ is an \emph{asymptotic form} of $f_S$. \remQED
\end{definition}

Thus, in the general form \eqref{eq:general_form_estimator} we will use $g_S = \Asymf$ when it is available and is a proper pdf. There are some technicalities for the cases when $\Asymf$ does not form a proper pdf which we defer from discussing in this work.
The estimator resulting from this radial approximation is
\begin{equation} \label{eq:asymptotic_estimator}
\hat{l}_{\mathrm{IS2}}(\gamma) := \frac{c_S \AsymFTail(\gamma)}{R} \sum_{r=1}^R \frac{ f_{(S,\bfTheta)}(S^{[r]}, \bfTheta^{[r]}) }{ c_S \Asymf( S^{[r]} )g_{\bfTheta \mid S}(\bfTheta^{[r]}\,|\, S^{[r]})}
\end{equation}
for $S^{[r]} \iidDist \AsymfTrun$ and $\bfTheta^{[r]} \indDist g_{\bfTheta \mid S}( \,\cdot\, | S^{[r]})$.

\begin{remark}
Define a ``correction factor'' to the asymptotic form, $\mathcal{R}(\gamma)$, by
$ \ell(\gamma)~=~\hat{\ell}_{\mathrm{Asym}}(\gamma) \mathcal{R}(\gamma) $; note that $\hat{\ell}_{\mathrm{Asym}}(\gamma) := c_S \AsymFTail(\gamma)$.
We can see that $\hat{\ell}_{\mathrm{IS2}}(\gamma)$ has a nice interpretation, because
\[ \hat{\ell}_{\mathrm{IS2}}(\gamma) = \hat{\ell}_{\mathrm{Asym}}(\gamma) \times \hat{\mathcal{R}}(\gamma) \,, \]
where $\hat{\mathcal{R}}(\gamma)$ is an unbiased Monte Carlo estimate of the factor $\mathcal{R}(\gamma)$. \remQED
\end{remark}

The recent applied probability literature has found the $\Asymf$ for a staggering array of distributions of $\bfX$. Perhaps the simplest case is when the $X_i$ are iid subexponential random variables. By definition (cf. \cite{foss2011introduction}), they satisfy
\begin{equation} \label{Subexponential_Asymptotics_IID}
f_S(s)\ \sim\ d \, f_1(s) \for{s \to \infty}\,.
\end{equation}
For sums of independent non-identically distributed (ind) subexponential variables (or for sums containing some subexponential and some lighter-tailed variables) we have
\begin{equation} \label{Subexponential_Asymptotics_Non_IID}
f_S(s)\ \sim\ \sum_{i=1}^d f_i(s)\ \sim\ \sum_{i\in I} f_i(s) \for{s \to \infty}
\end{equation}
where $I$ is the set of indices of slowest tail decay. The asymptotics in \eqref{Subexponential_Asymptotics_Non_IID} also hold in many regimes where dependence has been introduced, cf. \cite{foss2010sums,wuthrich2003asymptotic,alink2004diversification,alink2007diversification}.

A distribution can satisfy a stronger property called \emph{regular variation} which implies subexponentiality and hence the asymptotics above. Examples of regularly varying distributions are Cauchy, Fr\'{e}chet, and Pareto distributions \cite{bingham1989regular}. The lognormal and heavy-tailed Weibull distributions are subexponential but not regularly varying.

The Weibull distribution is interesting as it is a family which can be heavy-tailed, light-tailed (the Rayleigh distribution is a special case), or on the boundary between these (i.e.\ the exponential distribution). The asymptotic form for the heavy-tailed Weibull sum is covered by \eqref{Subexponential_Asymptotics_IID} and \eqref{Subexponential_Asymptotics_Non_IID} as the summands are subexponential. The difficulty in finding the asymptotics for the light-tailed case led the authors to investigate it in detail, leading to the paper \cite{asmussen2017tail} which uses results originally from \cite{balkema1993densities}.

\begin{proposition} \label{prop:light_weibull} Assume that $X_1, \dots, X_d$ are iid light-tailed $\mathsf{Weibull}(\beta, \lambda)$ where $\beta>1$, $\lambda \in \RL_+$, $d \ge 2$. Then
\begin{equation*}
c_S\AsymFTail(s)\ \sim\ \Bigl[ \frac{2\beta\pi}{\beta-1}\Bigr]^{(d-1)/2} d^{-1/2} \Big( \frac{s}{\lambda d} \Big)^{\beta(d-1)/2} \Ftail \Big( \frac{s}{d} \Big)^d
\end{equation*}
for $s \to \infty$.
\end{proposition}

The exposition in \cite{asmussen2017tail} details this and more general asymptotics (i.e.\ the independent but non-identically distributed case, and when the variables are not exactly Weibull but are `Weibull-like').

By its very construction, one would expect the estimator utilising an asymptotic form for the right-tail, \eqref{eq:asymptotic_estimator}, to
enjoy good efficiency properties as $\gamma\to\infty$.  As mentioned in the introduction, our goal is to provide a practically useful
estimator for `moderately' rare problems, in the range of $\gamma$ before the asymptotic regime takes hold.
As such, it is our view that the orthodox notions of efficiency have little meaning in our setting.
Nevertheless, we note that it is straightforward to verify that if the
ratio $f_{\bfTheta \mid \gamma}/g_{\bfTheta \mid \gamma}$ remains bounded by some finite constant $K\geq 1$ for all $\bftheta$ as $\gamma\to\infty$, then the estimator \eqref{eq:asymptotic_estimator} enjoys
bounded relative error, and if $K=1$ then the estimator enjoys vanishing relative error.

\subsection{The angular approximation}

The choice of angular approximation is not as obvious as was the choice of radial approximation.
Finding a conditional density $g_{\bfTheta \mid S}$ which is similar to $f_{\bfTheta \mid S}$ has little precedent in the literature.

Instead of taking an $S$ which is larger than $\gamma$ and asking `what is the distribution of $\bfTheta$ given this $S$?', we can instead ask `what is the distribution of $\bfTheta$ given $S > \gamma$?'.
This is the same situation that is studied in multivariate extreme value theory, where the spectral density characterises the behaviour of $f_{\bfTheta \mid S>\gamma}$ in
the limit as $\gamma\to\infty$ \cite{deHaanResnick1977}.
This second conditional distribution will resemble the first in cases that $\Exp[ S-\gamma \mid S > \gamma]$ converges quickly to zero when $\gamma$ becomes large.
Moreover, we have a computation benefit to finding a $g_{\bfTheta \mid S > \gamma}$ which is similar to $f_{\bfTheta \mid S > \gamma}$ as this distribution will be constant across all Monte Carlo iterates, in contrast to $g_{\bfTheta \mid S^{[r]}}$ and $f_{\bfTheta \mid S^{[r]}}$.

Nevertheless, when it is possible, we follow the same approach as the radial approximation and utilise some asymptotic information. However, we note that if one simply re-uses the previous asymptotic form, that is
\[ f_{\bfTheta \mid S}(\bftheta | s) = \frac{f_{S,\bfTheta}(s, \bftheta)}{f_S(s)} \sim \frac{f_{\bfX}(s \bftheta) |s|^{d-1}}{\Asymf(s)} =: g_{\bfTheta \mid S}(\bftheta | s) \,, \]
which may appear natural, then the estimator \eqref{eq:asymptotic_estimator} degenerates to the deterministic
\begin{equation*} \label{eq:double_asymptotic_estimator}
\hat{l}_{\mathrm{IS2}}(\gamma) := \frac{c_S \AsymFTail(\gamma)}{R} \sum_{r=1}^R 1 = c_S \AsymFTail(\gamma) \,.
\end{equation*}
Moreover, if it is known that $f_{\bfTheta \mid S > \gamma}$ degenerates in the limit [e.g.\ to a point mass at $1/d$ in each coordinate (perfect extremal dependence) or is degenerate on axes (independence in the extreme)],
this does not tell us what we should do for finite $\gamma$.

Indeed, when the summands are iid subexponentials, then the distribution of $(\bfTheta \mid S = s)$ as $s\to\infty$ degenerates to a discrete distribution over the $d$-dimensional unit vectors $\bfe_1$, \dots, $\bfe_d$
(with $\bfe_i$ having a single 1 in the $i$-th coordinate and zeros in all other coordinates).
This is just a re-casting of the principle of the single big jump (cf.\ \cite{foss2011introduction}).
Unfortunately, for finite $\gamma$ we cannot use this directly as in this case the likelihood ratio appearing in \eqref{eq:general_form_estimator} is not well defined.
One density, which we call the \emph{optimistic density} (see the algorithm below), that is not degenerate (and therefore will have well-defined likelihood ratio) but is
asymptotically equivalent to $(\bfTheta \mid S = s)$ for the case of ind subexponential summands is
\begin{align} \label{OptimisticPDF}
g_{\bfTheta \mid S}(\bftheta \mid s)
&= |s|^{d-1} \sum_{i=1}^d p_i(s) f_{\bfX_{-i}}(s \bftheta_{-i}) \ind{ \theta_i = 1-\bfone \cdot \bftheta_{-i} }
\end{align}
where $\bfX_{-i}$ and $\bftheta_{-i}$ denote the $(d-1)$-dimensional vectors  obtained from
$\bfX$ and $\bftheta$ by removing the elements in $i$-th coordinates ($X_i$ and $\theta_i$, respectively), and
the $p_i$ functions are defined by
\begin{equation} \label{optim_pi}
p_i(s) = \frac{\Ftail _i(s)}{\sum_{j=1}^d \Ftail _j(s)} \,.
\end{equation}
Algorithm~\ref{Alg:Optimism} shows a method for sampling from this $g_{\bfTheta \mid S}(\bftheta \mid s)$, and Proposition~\ref{prop:optimistic_asymptotic} shows that it
has limiting distribution consistent with \eqref{Subexponential_Asymptotics_Non_IID} as $s \to \infty$.

\begin{algorithm}[h]
\caption{Sampling from the optimistic angular density}\label{Alg:Optimism}
\begin{algorithmic}[1]
\Procedure{Optimistic}{$s$, $F_1$, \dots, $F_d$}
\State Simulate index $I$ in $\{1, \dots, d\}$ by $\Prob(I=i) = p_i(s)$ from \eqref{optim_pi}.
\For{$i = 1$ \textbf{to} $d$ \textbf{except} $I$}
\State $X_i^* \gets$ Random sample from $F_i$
\EndFor
\State $X_I^* \gets s - \sum_{i \not= I} X_i$ \Comment{This can be negative, but we are optimistic}
\State \textbf{return} $\bfTheta \gets \bfX^* / s$
\EndProcedure
\end{algorithmic}
\end{algorithm}

When the subexponential summands are only independent in the extreme, a simple
generalisation of Algorithm~\ref{Alg:Optimism} is to replace Lines~3 to~5 with
taking a random sample $\bfX^*$ from $f_{\bfX}$.

\begin{proposition} \label{prop:optimistic_asymptotic}

The optimistic density \eqref{OptimisticPDF} converges as $s \to \infty$ to the singular density
\begin{equation} \label{OptimisticLimitingDensity}
g_\infty(\bftheta) := \sum_{i=1}^{d} p_i \, \ind { \bftheta = \bfe_i } \,,
\end{equation}
where $p_i = \lim_{s\to\infty} p_i(s)$ for $i=1,\dots,d$.
\end{proposition}

\begin{proof}
For some $\bft=(t_1,\dots,t_d)' \in \mathbb{R}^d$, the characteristic function of $g_{\bfTheta \mid  S}$ is
\[
\phi_{g_{\bfTheta \mid  S}}(\bft \mid s)
= \Exp \exp\left(\ih\,\bft^\top \bfTheta \right)
= \Exp \left[ \exp\left(\ih\,\frac{\bft}{s}^\top \bfX^* \right) \right]
\]
where $\bfX^* = s \bfTheta$ as in Algorithm~\ref{Alg:Optimism}.

So, with $I$ as the discrete variable defined in Algorithm~\ref{Alg:Optimism}, we have
\begin{align*}
&\hphantom{=}~\phi_{g_{\bfTheta \mid  S}}(\bft \mid s) \\
&= \sum_{j=1}^d p_j(s) \Exp \left[ \e^{ \ih\,\frac{\bft}{s}^\top \bfX^* } \,\Big|\, I = j \right] \\
&= \sum_{j=1}^d p_j(s) \Exp \left[ \e^{ \ih \left[ \frac{\bft_{-j}}{s}^\top \bfX^*_{-j} + \frac{t_j}{s} (s - \bfone^\top \bfX^*_{-j}) \right] } \,\Big|\, I = j \right] \\
&= \sum_{j=1}^d p_j(s) \e^{\ih t_j} \Exp \left[ \e^{ \ih\,\frac{(\bft_{-j} - t_j \bfone)}{s}^\top \bfX_{-j} } \right] \\
&= \sum_{j=1}^d p_j(s) \e^{\ih t_j} \phi_{\bfX_{-j}} \left(  \frac{(\bft_{-j} - t_j \bfone)}{s} \right) \,.
\end{align*}
Therefore,
\[
\lim_{s \to \infty} \phi_{g_{\bfTheta | S}}(\bft ; s)
= \sum_{j=1}^d p_j \e^{\ih t_j}
= \sum_{j=1}^d p_j \e^{\ih \bft^\top \bfe_j}
\]
which corresponds to the singular density as in \eqref{OptimisticLimitingDensity}.
\end{proof}

\begin{remark}
The polar estimator for ind subexponential summands with the optimistic angular approximation \eqref{OptimisticPDF} simplifies to
\begin{align*}
&\hphantom{=}~\hat{l}_{\mathrm{IS2}}(\gamma) \\
&= \frac{c_S \AsymFTail(\gamma)}{R} \sum_{r=1}^R \frac{\mathrm{HM}( f_{X_1}( S^{[r]} \Theta^{[r]}_1), \dots, f_{X_d}( S^{[r]} \Theta^{[r]}_d) )}{c_S \Asymf( S^{[r]} )}
\end{align*}
where $S^{[r]} \iidDist \AsymfTrun$, $\bfTheta^{[r]} \indDist g_{\bfTheta \mid S}( \,\cdot\, | S^{[r]})$, and $\mathrm{HM}(\dots)$ is the harmonic mean of the inputs. \remQED
\end{remark}

The conditional angular asymptotic distribution is more challenging to obtain in the case of light-tailed summands. The following example shows these distributions differ qualitatively when different copulas are considered.

\begin{example}
Consider $X_1$ and $X_2$ to be \textsf{Exp}(1) variables which are: i) independent, ii) \textsf{Clayton}(1) dependent, or iii) \textsf{Ali-Mikhail-Haq}(-1) dependent. The sum densities can be calculated explicitly and
are given by
\[ f_S^{\mathrm{Ind}}(s) = s \e^{-s} \quad \text{for } s > 0 \,, \]
\[ f_S^{\mathrm{Cla}}(s) = \frac{2 - 2 \cosh(s) + s \sinh(s)}{(\cosh(s)-1)^2} \quad \text{for } s > 0 \,, \]
\[ f_S^{\mathrm{AMH}}(s) = 8 \csch(s)^3 \sinh(s/2)^4 \quad \text{for } s > 0 \,, \]
respectively.
Hence, for $s>0$ and $\theta \in (0,1)$, we have angular densities
\[ f^{\mathrm{Ind}}_{\Theta_1 | S}(\theta | s) = 1 \,, \]
\[ f^{\mathrm{Cla}}_{\Theta_1 | S}(\theta | s) = \frac{ s \e^{-s \theta}(\e^s - \e^{s \theta})(\e^{s \theta}-1)}{2 + s - 2\e^s + s \e^s} \,, \]
\[ f^{\mathrm{AMH}}_{\Theta_1 | S}(\theta | s) = \frac{ s \e^{-s \theta}(\e^s + \e^{2 s \theta})}{2(\e^s - 1)}  \,, \]
respectively. It is interesting to note that the asymptotic independence of the Clayton copula
would indicate that $f^{\mathrm{Cla}}_{\Theta_1 | S}(\theta | s) / f^{\mathrm{Ind}}_{\Theta_1 | S}(\theta | s) \to 1$ as $s\to\infty$ which is indeed the case. In contrast, $f^{\mathrm{AMH}}_{\Theta_1 | S}(\theta | s)$ degenerates to a pair of atoms at 0 and 1 as $s \to \infty$.  \remQED
\end{example}

One (light-tailed) case where we can determine an asymptotic angular distribution is for light-tailed Weibull sums. The angular asymptotic can be extracted from the results in \cite{asmussen2017tail},
and appears as follows.

\begin{proposition} \label{prop:light_weibull_angles}
Say $X_1, \dots, X_d$ are iid and are distributed as light-tailed $\mathsf{Weibull}(\beta, 1)$ with survival function $\Ftail(x)=\e^{-x^\beta}$ where $\beta>1$, $d \ge 2$.
Define the vector function $\bfW(x)$ component-wise by
\[ W_i(x) = \omega(x) ( X_i - x/d) \,, \quad \text{ for } i=1,\dots,d\,, \]
where $\omega(x) := \sqrt{2  \beta (\beta-1) (x/d)^{\beta-2}}$.
Then as $\gamma \to \infty$ we have
\[ (\bfW(\gamma) \mid S > \gamma) \convdistr \NormDist\left(\bfzero, (1 - \rho) \bfI + \rho \right) \,,\]
where $\rho = -1/(d-1)$.
\end{proposition}

Note that the $d$-dimensional multivariate normal distribution appearing in Proposition~\ref{prop:light_weibull_angles} is supported on a $(d-1)$-dimensional subspace.

When the asymptotic angular approximation is unavailable, there are several conceivable alternatives.
One can select a $g_{\bfTheta \mid S}$ from some family of distributions which has the appropriate support.
For instance, if $\bfX$ has non-negative components, then the support of $g_{\bfTheta \mid S}$ is the simplex
$ \mathbb{S}^{d-1} = \{ \bftheta \in \RL_+^d \, : \, \bftheta^\top\bfone = 1 \}$.
To the authors' knowledge, the only commonly known distribution over $\mathbb{S}^{d-1}$ is the Dirichlet distribution, and appears to be a natural candidate.

In some experiments carried out while writing this paper, we sampled $(\bfTheta \mid S > \gamma)$ using MCMC, then used the maximum likelihood Dirichlet fit to the samples as an angular approximation in the polar estimator.
Unfortunately the results were disappointing and are omitted --- the Dirichlet distribution struggles to fit the multimodal angular distributions which are characteristic of subexponential sums conditioned on taking large values.
We also attempted the MCMC flavour of the cross-entropy method as outlined by Chan and Kroese \cite{chan2012improved}, though the multimodality led to extremely high variance estimates (relative to the much simpler Asmussen--Kroese method).

We also performed an approximation of the angular density using Bernstein polynomials.
The angular density $f_{\bfTheta \mid S}(\bftheta \mid s) \propto f_{\bfX}(s \bftheta)$, so it is easy to calculate quantities which are proportional to the desired conditional density.
Using Bernstein polynomials effectively constructed an approximation which was a mixture of Dirichlet distributions using these unnormalised angular density values.
The results for these experiments are also omitted, since the number of mixture components required to create an accurate approximation easily becomes prohibitively large (then, the computation time for evaluating the pdf of the mixture becomes a bottleneck).

\section{Results} \label{Sec:Results}

In this section we give illustrative results of numerical experiments.
For subexponential summands, we compare to the most competitive alternative, the Asmussen--Kroese
estimator, and for light-tailed summands we compare to the standard IS approach of exponential tilting.
In what follows, we adopt Mathematica's parameterisations for the lognormal, Pareto, and Weibull distributions. The code we used is available online \cite{PolarCode}.

\subsection{Subexponential Summands}

Below we present the estimates and the estimated relative errors for the polar estimator and the Asmussen--Kroese estimator for various distributions of $\bfX$. Each estimator is given $R = 10^5$ iid samples of $\bfX$.

The first test (Figs.\ \ref{test1:ests} and \ref{test1:relerrs}) takes the sum of $d=12$ independent lognormal random variables, with marginal distributions $X_i~\sim~\LNDist(-\frac{i}{d}, \sqrt{\frac{i}{d}})$. Here, the sum behaves asymptotically as the dominant term $X_{12} \sim \LNDist(-1, 1)$, and the optimistic angular distribution is used.

\begin{figure}[h]
	\centering
	\includegraphics[width=0.49\textwidth]{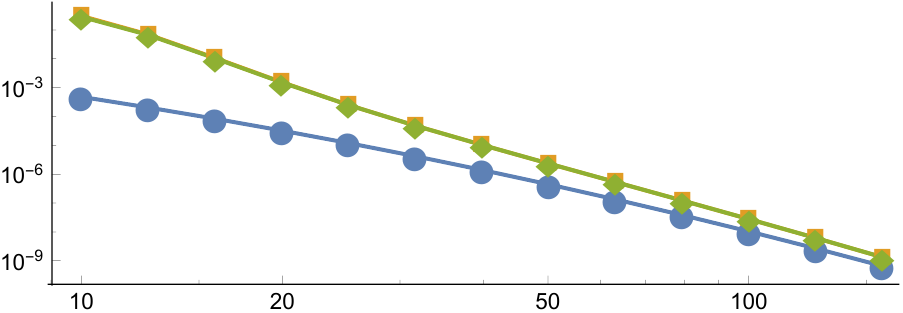} \\
	\includegraphics[height=1em]{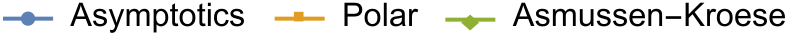}
	\caption{Estimates of $\Prob(S > \gamma)$ from each estimator.}
	\label{test1:ests}
\end{figure}

\begin{figure}[h]
	\centering
	\includegraphics[width=0.49\textwidth]{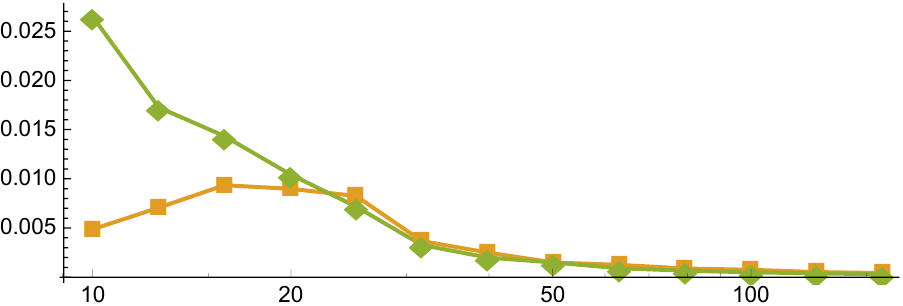} \\
	\includegraphics[height=0.9em]{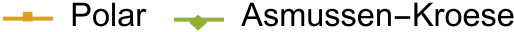}
	\caption{Estimated relative errors for each estimator.}
	\label{test1:relerrs}
\end{figure}

The second test (Figs.\ \ref{test2:ests} and \ref{test2:relerrs}) considers the sum of $d=16$ independent Pareto random variables, where $X_i \sim \ParetoDist(1, i, 0)$. The sum behaves asymptotically as the dominant term $X_1 \sim \ParetoDist(1, 1, 0)$, and the optimistic angular distribution is used.

\begin{figure}[h]
	\centering
	\includegraphics[width=0.49\textwidth]{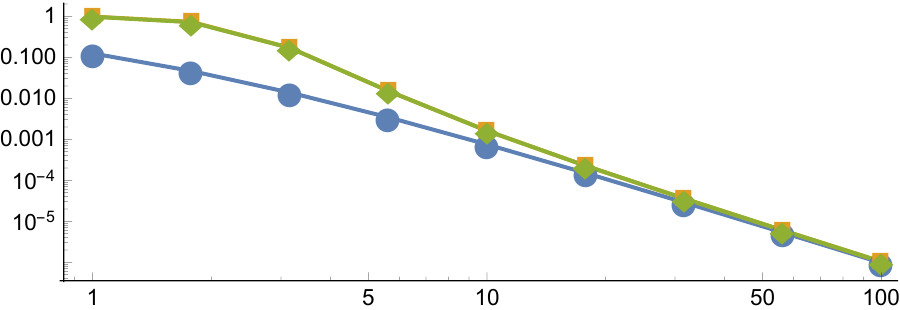} \\
	\includegraphics[height=1em]{legend1a.pdf}
	\caption{Estimates of $\Prob(S > \gamma)$ from each estimator.}
	\label{test2:ests}
\end{figure}

\begin{figure}[h]
	\centering
	\includegraphics[width=0.49\textwidth]{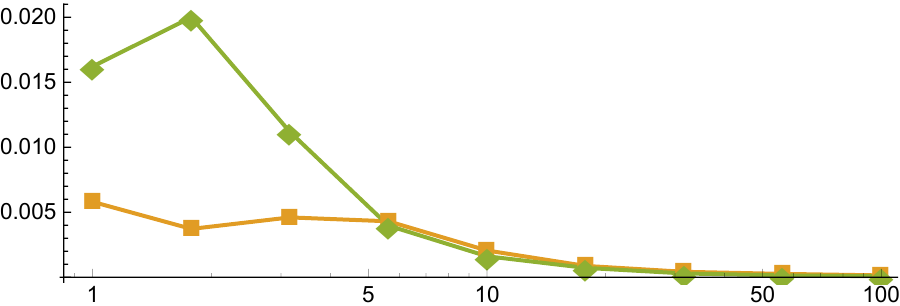} \\
	\includegraphics[height=0.9em]{legend1b.pdf}
	\caption{Estimated relative errors for each estimator.}
	\label{test2:relerrs}
\end{figure}

The third test (Figs.\ \ref{test3:ests} and \ref{test3:relerrs}) considers the sum of $d=8$ independent heavy-tailed Weibull variables. The marginal distributions are $X_i~\sim~\WeibullDist(\frac{1}{4}, \frac{i}{d})$. The sum behaves asymptotically as the last summand $X_8~\sim~\WeibullDist(\frac{1}{4}, 1)$, and the optimistic angular distribution is used.

\begin{figure}[h]
	\centering
	\includegraphics[width=0.49\textwidth]{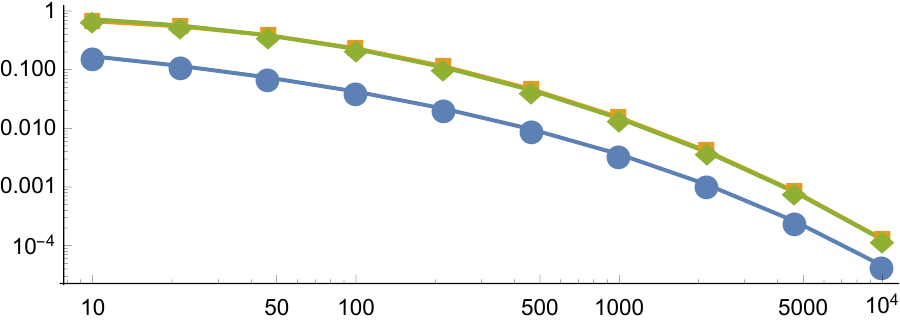} \\
	\includegraphics[height=1.0em]{legend1a.pdf}
	\caption{Estimates of $\Prob(S > \gamma)$ from each estimator.}
	\label{test3:ests}
\end{figure}

\begin{figure}[h]
	\centering
	\includegraphics[width=0.49\textwidth]{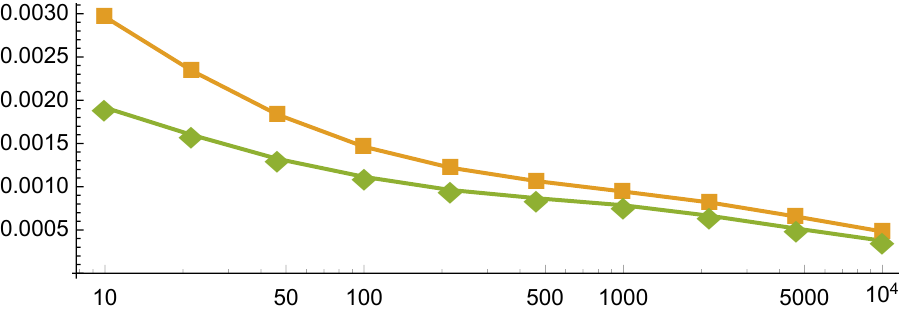} \\
	\includegraphics[height=0.9em]{legend1b.pdf}
	\caption{Estimated relative errors for each estimator.}
	\label{test3:relerrs}
\end{figure}

\subsection{Light-tailed Weibull Summands}

This fourth test (Figs.\ \ref{test4:ests} and \ref{test4:relerrs}) takes the sum of $d=10$ iid light-tailed Weibulls, where $X_i \sim \WeibullDist(2, 1)$. An asymptotic survival function for the sum is given by Proposition~\ref{prop:light_weibull}, and the optimistic angular distribution used is from Proposition~\ref{prop:light_weibull_angles}. Instead of the Asmussen--Kroese method, which is designed for subexponential summands, we have compared the polar estimator against \emph{exponential tilting}. The exponential tilting method is usually easy to implement but it takes some effort in this situation. Simulating each exponentially tilted Weibull variable is done via acceptance--rejection; the proposals come from the gamma distribution which is moment-matched to the asymptotic normal approximation for the exponentially tilted Weibull distribution, cf.\ Section 6 of \cite{asmussen2017tail}.

\begin{figure}[h]
	\centering
	\includegraphics[width=0.49\textwidth]{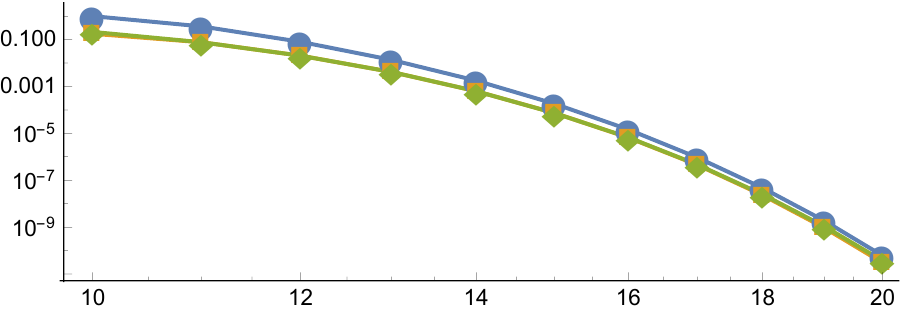} \\
	\includegraphics[height=1.0em]{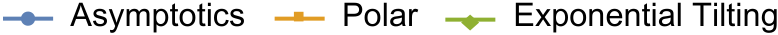}
	\caption{Estimates of $\Prob(S > \gamma)$ from each estimator.}
	\label{test4:ests}
\end{figure}

\begin{figure}[h]
	\centering
	\includegraphics[width=0.49\textwidth]{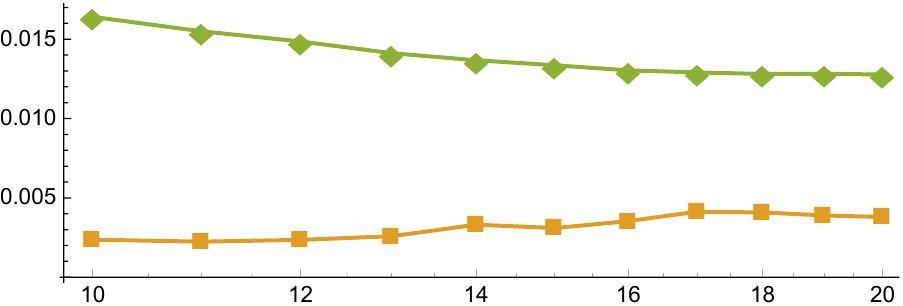} \\
	\includegraphics[height=1.0em]{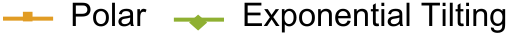}
	\caption{Estimated relative errors for each estimator.}
	\label{test4:relerrs}
\end{figure}

\subsection{Dependent Summands}

Next we reproduce the three subexponential tests above with dependence added by Archimedean copulas. We use the Asmussen--Kroese estimator as outlined in Section~3.2.2.2 of \cite{nandayapa2008risk} as the traditional form of this estimator needs to be adapted for the case of dependent summands.

The fifth test (Figs.\ \ref{test5:ests} and \ref{test5:relerrs}) recreates the first test, with $d=12$ lognormal variables and marginals $X_i~\sim~\LNDist(-\frac{i}{d}, \sqrt{\frac{i}{d}})$, except dependence is added via a $\FrankCop(\frac{1}{2})$ copula.

\begin{figure}[h]
	\centering
	\includegraphics[width=0.49\textwidth]{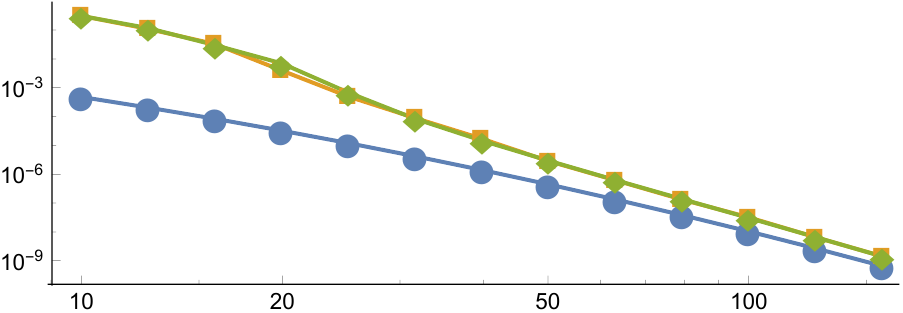} \\
	\includegraphics[height=1.0em]{legend1a.pdf}
	\caption{Estimates of $\Prob(S > \gamma)$ from each estimator.}
	\label{test5:ests}
\end{figure}

\begin{figure}[h]
	\centering
	\includegraphics[width=0.49\textwidth]{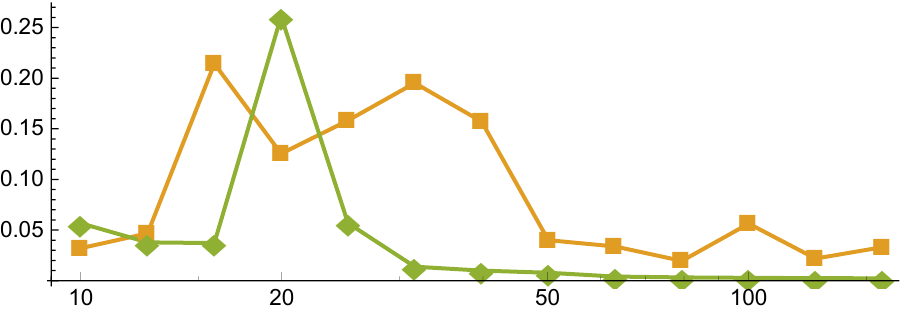} \\
	\includegraphics[height=0.9em]{legend1b.pdf}
	\caption{Estimated relative errors for each estimator.}
	\label{test5:relerrs}
\end{figure}

We see immediately that introducing even a mild level of dependence gives rise to substantially more variability in the polar estimator in the pre-asymptotic regime as
a result of using the optimistic angular distribution --- an illustration of likelihood ratio degeneracy in $d$.
To illustrate this point, we carry out the same test with $d=4$ in Figs.\ \ref{test6:ests} and \ref{test6:relerrs}.

\begin{figure}[h]
	\centering
	\includegraphics[width=0.49\textwidth]{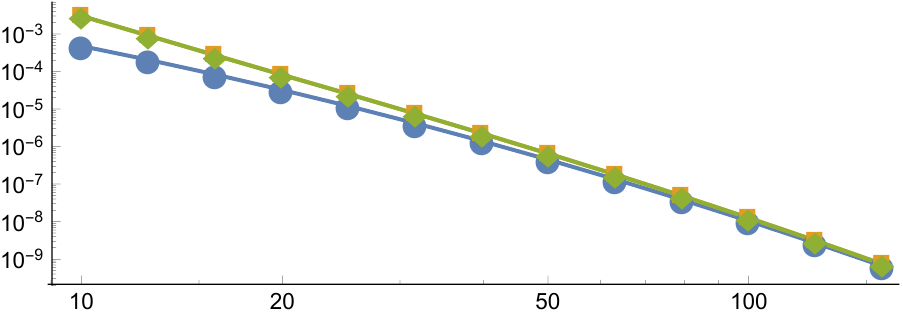} \\
	\includegraphics[height=1.0em]{legend1a.pdf}
	\caption{Estimates of $\Prob(S > \gamma)$ from each estimator.}
	\label{test6:ests}
\end{figure}

\begin{figure}[h]
	\centering
	\includegraphics[width=0.49\textwidth]{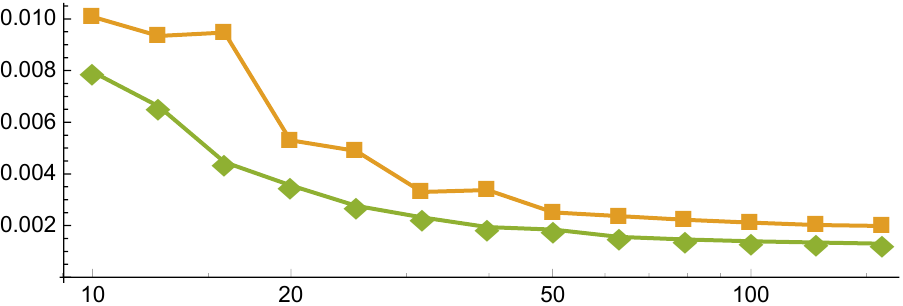} \\
	\includegraphics[height=0.9em]{legend1b.pdf}
	\caption{Estimated relative errors for each estimator.}
	\label{test6:relerrs}
\end{figure}

The seventh test (Figs.\ \ref{test7:ests} and \ref{test7:relerrs}) is similar to the second test above, considering $d=16$ Pareto random variables where $X_i \sim \ParetoDist(1, i, 0)$, except that the summands exhibit dependence via a  $\ClaytonCop(\frac{9}{10})$ copula.

\begin{figure}[h]
	\centering
	\includegraphics[width=0.49\textwidth]{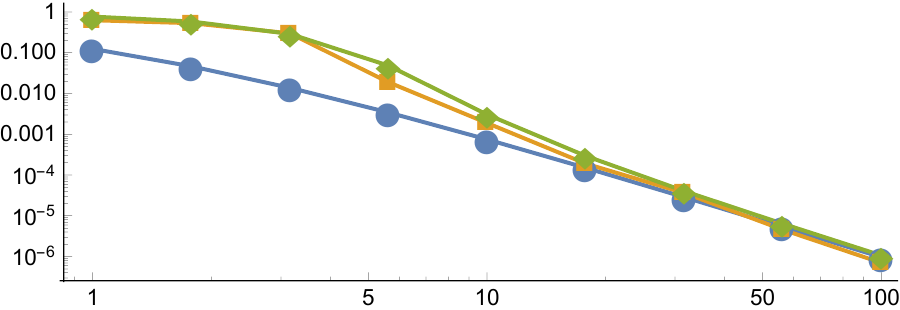} \\
	\includegraphics[height=1.0em]{legend1a.pdf}
	\caption{Estimates of $\Prob(S > \gamma)$ from each estimator.}
	\label{test7:ests}
\end{figure}

\begin{figure}[h]
	\centering
	\includegraphics[width=0.49\textwidth]{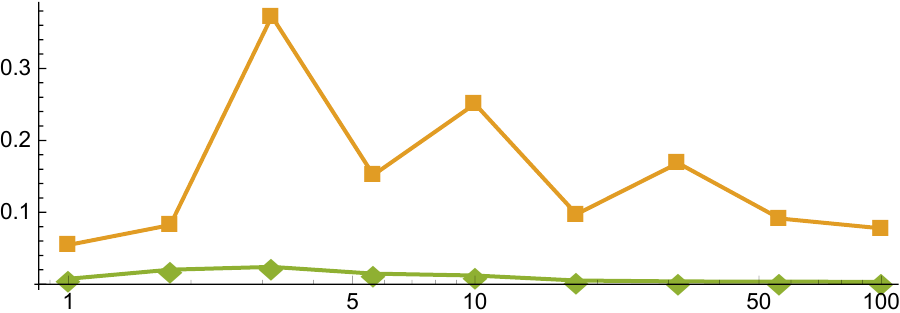} \\
	\includegraphics[height=0.9em]{legend1b.pdf}
	\caption{Estimated relative errors for each estimator.}
	\label{test7:relerrs}
\end{figure}

Once again, we observe significant likelihood ratio degeneracy in the polar estimator, and for comparison we repeat the experiment with $d=4$ in Figs.\ \ref{test8:ests} and \ref{test8:relerrs}.

\begin{figure}[h]
	\centering
	\includegraphics[width=0.49\textwidth]{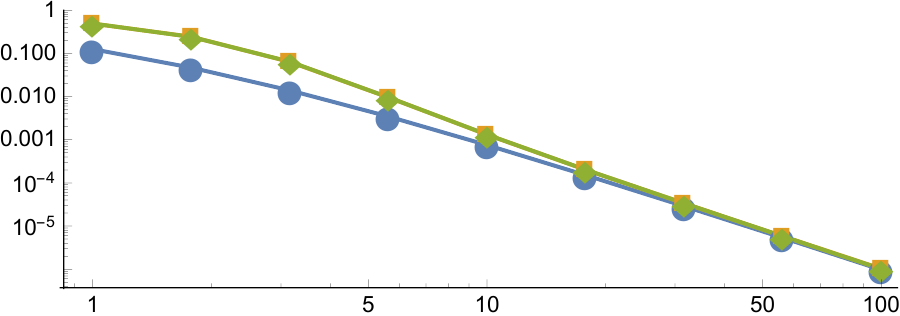} \\
	\includegraphics[height=1.0em]{legend1a.pdf}
	\caption{Estimates of $\Prob(S > \gamma)$ from each estimator.}
	\label{test8:ests}
\end{figure}

\begin{figure}[h]
	\centering
	\includegraphics[width=0.49\textwidth]{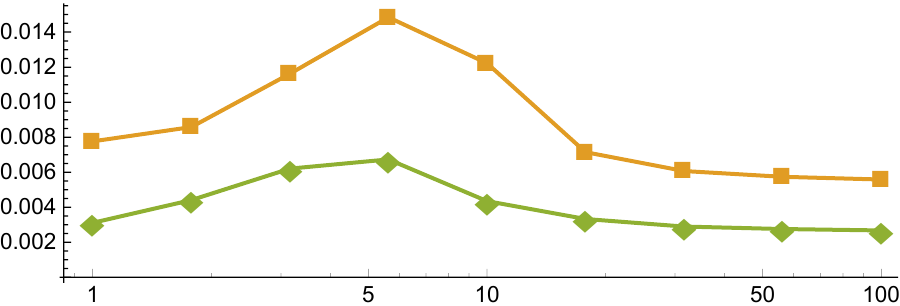} \\
	\includegraphics[height=0.9em]{legend1b.pdf}
	\caption{Estimated relative errors for each estimator.}
	\label{test8:relerrs}
\end{figure}

Test nine (Figs.\ \ref{test9:ests} and \ref{test9:relerrs}) is similar to the third test, with $d=8$ heavy-tailed Weibull variables with marginal distributions $X_i~\sim~\WeibullDist(\frac{1}{4}, \frac{i}{d})$, except with a $\GumbelCop(1.25)$ copula (which is dependent in the extreme).

\begin{figure}[h]
	\centering
	\includegraphics[width=0.49\textwidth]{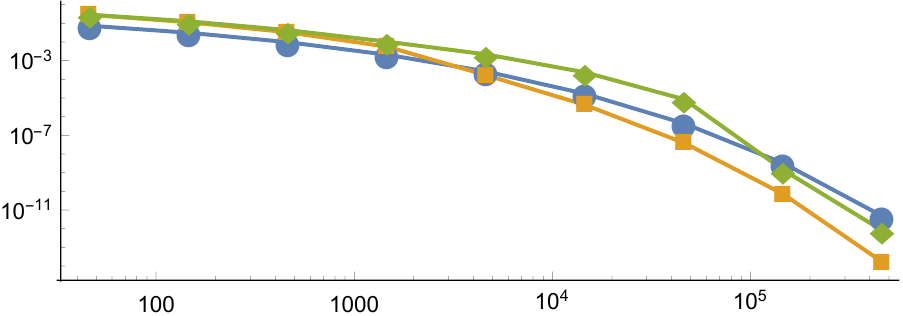} \\
	\includegraphics[height=1.0em]{legend1a.pdf}
	\caption{Estimates of $\Prob(S > \gamma)$ from each estimator.}
	\label{test9:ests}
\end{figure}

\begin{figure}[h]
	\centering
	\includegraphics[width=0.49\textwidth]{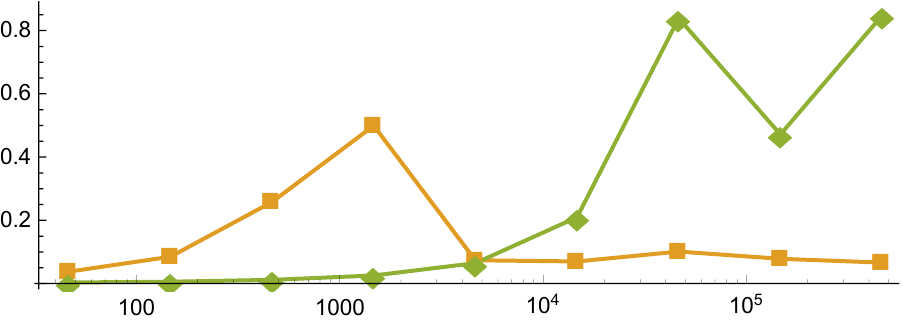} \\
	\includegraphics[height=0.9em]{legend1b.pdf}
	\caption{Estimated relative errors for each estimator.}
	\label{test9:relerrs}
\end{figure}

As one would expect, as $\gamma$ increases, both estimators behave increasingly poorly due to the upper-tail dependence of this copula.
We repeat the experiment in Figs.\ \ref{test10:ests} and \ref{test10:relerrs} with $d=2$ to more clearly illustrate this phenomenon.

\begin{figure}[h]
	\centering
	\includegraphics[width=0.49\textwidth]{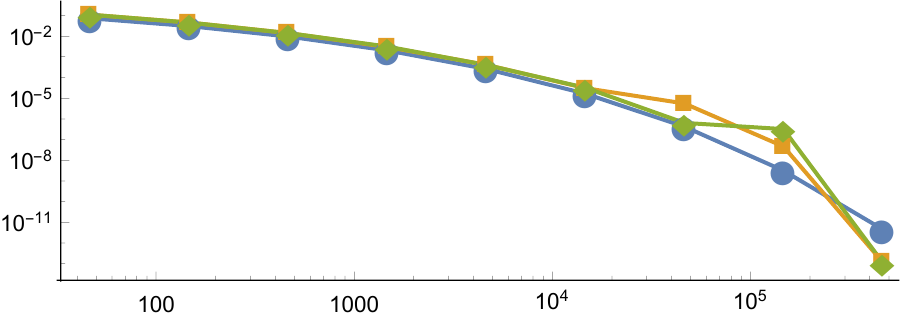} \\
	\includegraphics[height=1.0em]{legend1a.pdf}
	\caption{Estimates of $\Prob(S > \gamma)$ from each estimator.}
	\label{test10:ests}
\end{figure}

\begin{figure}[h]
	\centering
	\includegraphics[width=0.49\textwidth]{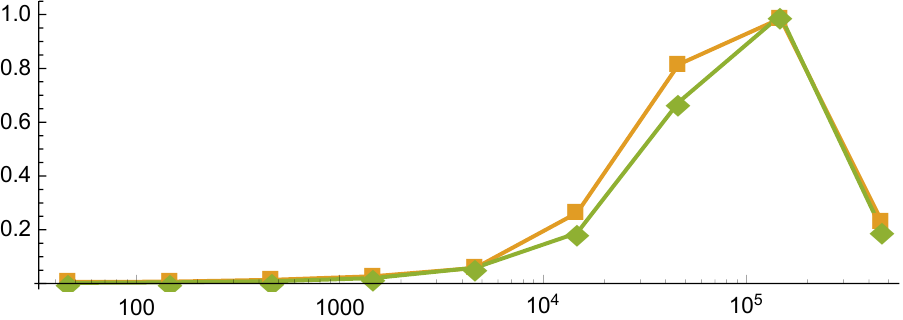} \\
	\includegraphics[height=0.9em]{legend1b.pdf}
	\caption{Estimated relative errors for each estimator.}
	\label{test10:relerrs}
\end{figure}

\section{Conclusion} \label{Sec:Conclusion}

On the tests carried out in this work, our estimator appears to perform on par with the Asmussen--Kroese method for independent subexponential summands, and outperforms all the other methods compared against (i.e.\ the improved cross-entropy method, fitting mixtures of Dirichlet variables, and Bernstein polynomial approximation).

Moreover, for the comparison of iid light-tailed Weibull summands, our estimator outperforms exponential tilting.

However, with the introduction of dependence for subexponential summands (even with upper-tail independence of the copula) the performance of our estimator
degrades rapidly as dimension increases (likelihood ratio degeneracy) as a consequence of utilising the optimistic angular distribution, and
unsurprisingly performs poorly when the copula has upper-tail dependence.
Thus there remains the opportunity for further research into suitable choice of the angular distribution in the case of dependent summands.

\begin{acknowledgements}

This work was supported under the Australian Research Council's Discovery Projects funding scheme (DP180101602). Also, PJL was supported by an Australian Government Research Training Program Scholarship and by the Australian Research Council Centre of Excellence for
Mathematical \& Statistical Frontiers (ACEMS), under grant number CE140100049.

\end{acknowledgements}

\end{document}